\documentclass[journal,draftcls,onecolumn,letterpaper,twoside,11pt]{IEEEtran}

\usepackage{cite}

\ifCLASSINFOpdf
  \usepackage[pdftex]{graphicx}
\else
\fi
\usepackage[cmex10]{amsmath}
\usepackage{array}
\ifCLASSOPTIONcompsoc
 \usepackage[caption=false,font=normalsize,labelfont=sf,textfont=sf]{subfig}
\else
  \usepackage[caption=false,font=footnotesize]{subfig}
\fi
\usepackage{fixltx2e}
\usepackage{dblfloatfix}

\ifCLASSOPTIONcaptionsoff
  \usepackage[nomarkers]{endfloat}
 \let\MYoriglatexcaption\caption
 \renewcommand{\caption}[2][\relax]{\MYoriglatexcaption[#2]{#2}}
\fi

\usepackage{url}

\usepackage{amsthm}
\usepackage{mathrsfs}
\usepackage{color}
\usepackage{mathtools,amssymb, nccmath}
\usepackage{bbm}

\theoremstyle{plain}
\newtheorem{theorem}{Theorem}

\newtheorem{lemma}{Lemma}
\newtheorem*{conjecture*}{Conjecture}
\theoremstyle{definition}
\newtheorem{remark}{Remark}

\newcommand{\ind}[1]{\mathbbm{1}_{\{#1\}}}   
\newcommand{\Exp}{{\mathsf{E}}}
\newcommand{\Pro}{{\mathsf{P}}}
\newcommand{\F}{\mathscr{F}}
\newcommand{\hth}{\hat{\theta}}
\newcommand{\hvth}{\hat{\vartheta}}
\newcommand{\C}{\mathscr{C}}
\newcommand{\T}{{T_{\rm o}}}
\newcommand{\J}{{\mathsf{J}}}
\newcommand{\V}{\mathscr{V}}

\begin{document}
\title{Optimal Stopping for Interval Estimation in Bernoulli Trials}

\author{Tony~Yaacoub,
        George~V.~Moustakides,~\IEEEmembership{Senior Member,~IEEE} 
        and Yajun~Mei
\thanks{T. Yaacoub and Y. Mei are with the H. Milton Stewart School of Industrial and Systems Engineering, Georgia Institute of Technology, Atlanta,
GA, USA. E-mail: (tyaacoub, ymei3)@gatech.edu.
Website: http://www2.isye.gatech.edu/\raisebox{0.1em}{{\tiny$\sim$}}ymei/.}
\thanks{G.\,V. Moustakides is with the Department of Computer Science, Rutgers University, Piscataway, NJ, USA and the Electrical and Computer Engineering Department, University of Patras, Rion, Greece. E-mail: george.moustakides@rutgers.edu, moustaki@upatras.gr.
Website: https://www.cs.rutgers.edu/\raisebox{0.1em}{{\tiny$\sim$}}gm463/.}
\thanks{Manuscript received ~~; revised ~~.}}

\maketitle

\begin{abstract}
We propose an optimal sequential methodology for obtaining confidence intervals for a binomial proportion $\theta$. Assuming that an i.i.d. random sequence of Benoulli($\theta$) trials is observed sequentially, we are interested in designing a)~a stopping time $T$ that will decide when is the best time to stop sampling the process, and b)~an optimum estimator $\hth_T$ that will provide the optimum center of the interval estimate of $\theta$. We follow a semi-Bayesian approach, where we assume that there exists a prior distribution for $\theta$, and our goal is to minimize the average number of samples while we guarantee a minimal coverage probability level. The solution is obtained by applying standard optimal stopping theory and computing the optimum pair $(T,\hth_T)$ numerically. Regarding the optimum stopping time component $T$, we demonstrate that it enjoys certain very uncommon characteristics not encountered in solutions of other classical optimal stopping problems. Finally, we compare our method with the optimum fixed-sample-size procedure but also with existing alternative sequential schemes.
\end{abstract}


\IEEEpeerreviewmaketitle

\section{Introduction}

\IEEEPARstart{I}{nterval} estimation of a binomial proportion $\theta$ is one of the most basic problems in statistics with many important real-world applications. Some classical applications include interval estimation of the prevalence of a rare disease \cite{sullivan13};  interval estimation of the overall response rate in clinical trials \cite{abramson13}; and accuracy assessment in remote sensing \cite{mk98}. In these applications, the sample size is fixed in advance, and a \textit{confidence interval} for $\theta$ is obtained. There exists extensive bibliography regarding derivations of confidence intervals for $\theta$ when the sample size is fixed. Perhaps, the most widely known in this category is Wald's interval, which takes the form $\hat{\theta}_\text{T}\pm z_{\frac{\alpha}{2}}\sqrt{\frac{\hat{\theta}_\text{T}(1-\hat{\theta}_\text{T})}{\text{T}}}$, where $\text{T}$ is the fixed sample size, $1-\alpha$ expresses the desired coverage probability, $\hat{\theta}_\text{T}$ is the sample mean of $\theta$
and $z_{\frac{\alpha}{2}}$ satisfies $Q(z_{\frac{\alpha}{2}}) =\frac{\alpha}{2}$ with $Q(x)$ denoting the complementary cdf of a standard $N(0,1)$ Gaussian random variable. This confidence interval is derived based on
the asymptotic normality of $\hat{\theta}_\text{T}$ and, therefore, exhibits poor behavior when $\text{T}\theta(1-\theta)$ is small \cite{ac98, bcd01,n98, v93}.
Several efforts to improve Wald's classical method are reported in \cite{w27,cp34,s54,c56,bs83,ac98}. 
There are also Baysian-based techniques \cite{bcd01, bcd02, r03} while in
\cite{ac98, bcd01, n98, pa08, v93} there exists interesting surveys that evaluate the relative performance of the above methods. Finally we must mention that \cite{t14} provides explicit formulas for the required sample size that can guarantee a prescribed coverage probability.

In many modern applications, sampling observations is costly and time consuming. Therefore, there is a desire to limit the sampling size without, however, compromising the quality of the interval estimate. For instance, in automatic fraud detection in finance, one needs to manually go through the ``suspect'' financial transactions that are automatically detected as fraudulent by some machine learning or other computer algorithm. Since the manual process is expensive in labor and cost, it is desirable to quickly estimate, with high confidence, what percentage of the suspect transactions are truly fraudulent. A different motivating application is in Statistical Model Checking, where with an approximate verification method, one overcomes the state space explosion problem for probabilistic systems by Monte Carlo simulations. Given an executable stochastic system, we verify a system's property with simulation and we desire to estimate the probability $\theta$ by which the system satisfies the property in question. The goal is to estimate $\theta$ within acceptable margins of error and confidence (see \cite{Jegourel17} and references therein). Because Monte Carlo simulations very often tend to require extensive time and computing power, it is advantageous to reduce their number assuring, at the same time, satisfactory quality levels for the corresponding estimate. The sequential version of the interval estimation aims exactly at reducing the sample size by selecting it to be random and, in particular, a stopping time controlled by the observations themselves. The literature focusing on the sequential setup of the problem is limited compared to its fixed sample-size counterpart (see \cite{cr65,t61,f10}). However, none of these articles is able to claim optimality of their corresponding schemes in any sense.

The objective of our current work is to offer optimum \textit{sequential} methods for interval estimation of $\theta$, with the quality of the estimate expressed through the \textit{coverage probability}. In addition to deriving the optimum scheme, we will also demonstrate some very uncommon but highly interesting properties of the optimum solution. These properties are not encountered in optimum sequential schemes derived for other well known sequential problems (i.e.~sequential hypothesis testing). We must also add that our methodology exhibits similarities with the work developed in \cite{c77}. However, the focus in \cite{c77} is on the actual estimate of $\theta$ with the adopted criterion being a variation of the classical mean square error.
In our work, as we pointed out, we focus on confidence intervals and coverage probabilities; and, as it turns out, this difference makes our derivations and proofs far more complicated, requiring original analytical methodology. This becomes particularly apparent when we attempt to establish the validity of the unique properties, mentioned before, that characterize our optimum solution.

The remainder of this article is organized as follows. In Section\,II we discuss our proposed framework for interval estimation for $\theta$ and propose a well-defined optimization problem and discuss its general solution. In Section\,III we focus on the computational aspects of the optimum scheme and the unique properties that they characterize it. In Section\,IV we compare the proposed scheme against the fixed-sample-size and two existing sequential methods in the literature. Finally, Section\,V contains our conclusions.

\section{Proposed Framework}

We observe sequentially an i.i.d. process $X_{1},X_2,\ldots$ of Bernoulli random variables with $X_t\in\{0,1\}$ and $\Pro(X_t = 1) = \theta = 1 - \Pro(X_t = 0), \theta\in[0,1]$. The goal is to provide a confidence interval for $\theta$. We are interested in confidence intervals of fixed width equal to $2h$ for some pre-specified $h\in(0,\frac{1}{2})$. We would also like our scheme to be able to guarantee a coverage probability equal to $1-\alpha$, where $\alpha\in(0,1)$ is given. Our scheme consists of a pair $(T,\hat{\theta}_T)$, that is, a stopping time $T$ and a \textit{mid-point estimator}\footnote{The estimate $\hat{\theta}_T$ does not have the meaning of a classical parameter estimator. It is the mid-point of the confidence interval $[\hat{\theta}_T-h,\hat{\theta}_T+h]$ and does not necessarily constitute an efficient estimate of $\theta$.} $\hat{\theta}_T$, where $T$ is adapted to the observation history (filtration generated by the observations) and $\hat{\theta}_T$ is a function of the observations accumulated up to the time of stopping $T$. We would like to solve the following constrained optimization problem for the optimum pair 
\begin{equation}
\inf_{T,\hat{\theta}_T}\Exp[T|\theta],~\text{subject to:}~\Pro(|\hat{\theta}_T-\theta|>h|\theta)\leq\alpha,
\label{eq:first_criterion}
\end{equation}
where the desired interval estimate is $[\hat{\theta}_T-h,\hat{\theta}_T+h]$ (with the two ends cropped at 0 and 1, respectively, whenever they exceed the two limits) and where $\Pro(\cdot|\theta)$ and $\Exp[\cdot|\theta]$ denote probability and expectation for \textit{given} $\theta$.

Although \eqref{eq:first_criterion} seems as the ideal formulation, it unfortunately targets an infeasible goal. We note that we are asking for the pair $(T,\hat{\theta}_T)$ to minimize the average number of samples \textit{for every value} of the parameter $\theta$. In other words, we want our scheme to enjoy a \textit{uniform} optimality property over all $\theta$, a requirement which is impossible to satisfy. In order to be able to find a solution that has a well-defined form of optimality, we adopt a semi-Bayesian approach\footnote{The term ``semi-Bayesian'' is used because our setup involves two different components where one is optimized while the other is constrained, unlike full-Bayesian approaches that combine all terms into a single performance measure.} and assume that a prior $\pi(\theta)$ for $\theta$ is available. This allows for the following modification of the previous constrained optimization
\begin{equation}
\inf_{T,\hat{\theta}_T}\Exp[T],~\text{subject to:}~\Pro(|\hat{\theta}_T-\theta|>h)\leq\alpha
\label{eq:second_criterion}
\end{equation}
where $\Pro(\cdot)$ and $\Exp[\cdot]$ denote probability and expectation including \textit{averaging over $\theta$} with the help of the prior.

\begin{remark} We must emphasize that the constraint in \eqref{eq:second_criterion} does not guarantee that the desired coverage probability will also hold for each individual $\theta$, namely $\Pro(|\hat{\theta}_T-\theta|>h|\theta)\leq\alpha$, a property which is particularly desirable in practice. Perhaps, a more meaningful problem to consider in place of \eqref{eq:second_criterion} would have been
\begin{equation}
\inf_{T,\hat{\theta}_T}\Exp[T],~\text{subject to:}~\sup_\theta\Pro(|\hat{\theta}_T-\theta|>h|\theta)\leq\alpha,
\label{eq:second_criterion*}
\end{equation}
that assures a coverage probability of at least $1-\alpha$ for \textit{every} $\theta$. Unfortunately, it is unclear how to derive the optimal solution to this alternative formulation. Consequently, we focus on \eqref{eq:second_criterion} as the optimum scheme we are going to develop, but in our numerical examples, we will evaluate it in terms of \eqref{eq:second_criterion*} as well.
\end{remark}

Let $c>0$ denote a Lagrange multiplier that we use to combine the two terms in \eqref{eq:second_criterion} into a single cost function $J(T, \hat \theta_{T})= c  \Exp[T] + \Pro(|\hat{\theta}_T-\theta|>h),$ and consider the \textit{unconstrained} optimization problem
\begin{equation}
\inf_{T,\hat{\theta}_T}J(T, \hat \theta_{T})
=\inf_{T,\hat{\theta}_T}\left\{c  \Exp[T] + \Pro(|\hat{\theta}_T-\theta|>h)\right\}.
\label{eq:third_criterion}
\end{equation}
We will first identify the solution to \eqref{eq:third_criterion} and then demonstrate that a proper selection of $c$ can also solve the constrained problem in \eqref{eq:second_criterion}.

\subsection{The Unconstrained Problem}

We start by considering the classical Bayes estimation problem for fixed sample size $t$
\begin{equation}
\inf_{\hat{\theta}_t}\Pro(|\hat{\theta}_t-\theta|>h).
\label{eq:Bayes}
\end{equation}
If we observe $\F_t=\sigma\{X_1,\ldots,X_t\}$ then, given that $\{X_t\}$ is i.i.d.~Bernoulli($\theta$), the probability to obtain a specific combination of samples given $\theta$ is equal to $\theta^{S_t}(1-\theta)^{t-S_t},$ where $S_t=\sum_{k=1}^tX_k$ is the number of ``successes'' up to time $t$.  This implies that the posterior probability density of $\theta$ given the observations can be written as
\begin{equation}
\pi_t(\theta|\F_t)=\pi_t(\theta|S_t)=\frac{\theta^{S_t}(1-\theta)^{t-S_t}\pi(\theta)}{\int_0^1\theta^{S_t}(1-\theta)^{t-S_t}\pi(\theta)\,d\theta}.
\label{eq:posterior}
\end{equation}
From Bayesian estimation theory \cite[Page 142]{p88}, we have that the optimization in \eqref{eq:Bayes} is achieved by the following Bayes estimator
\begin{equation}
\hat{\vartheta}_t(S_t)=\text{arg}\inf_{\hat{\theta}_t}\Pro(|\hat{\theta}_t-\theta|>h|\F_t)=\text{arg}\sup_{\hat{\theta}_t}\int_{\max\{\hat{\theta}_t-h,0\}}^{\min\{\hat{\theta}_t+h,1\}}\pi_t(\theta|S_t)\,d\theta,
\label{eq:1}
\end{equation}
yielding the corresponding optimum conditional complementary coverage probability 
\begin{align}
\begin{split}
\C_t(S_t)&=\inf_{\hat{\theta}_t}\Pro(|\hat{\theta}_t-\theta|>h|\F_t)=1-\sup_{\hat{\theta}_t}\int_{\max\{\hat{\theta}_t-h,0\}}^{\min\{\hat{\theta}_t+h,1\}}\pi_t(\theta|S_t)\,d\theta\\
&=1-\int_{\max\{\hat{\vartheta}_t(S_t)-h,0\}}^{\min\{\hat{\vartheta}_t(S_t)+h,1\}}\pi_t(\theta|S_t)\,d\theta.
\end{split}
\label{eq:2}
\end{align}
From \eqref{eq:1} and \eqref{eq:2} we observe that both quantities $\hat{\vartheta}_t(S_t),\C_t(S_t)$ are $\F_t$-measurable and, more precisely, functions of $S_t$. For known prior $\pi(\theta)$, we can, at least numerically, compute the Bayes estimate and the corresponding optimum conditional complementary coverage probability for each combination of integer pair $(t,S_t)$. 
\begin{remark}
By focusing on \eqref{eq:1}, we can make a small but interesting observation: Regarding the Bayes estimate $\hat{\vartheta}_t(S_t)$ it is easy to verify that
\begin{equation}
h\leq\hat{\vartheta}_t(S_t)\leq1-h.
\end{equation}
Indeed, this is clear, because if in \eqref{eq:1} we select $\hat{\theta}_t<h$ or $\hat{\theta}_t>1-h$, this will yield an inferior cost compared to the selection $\hat{\theta}_t=h$ or $\hat{\theta}_t=1-h$, respectively. The implication of this observation is that $\hat{\vartheta}_t(S_t)$ will be \textit{biased and inconsistent} when considered as an estimate of the true parameter $\theta$, at least for values of $\theta$ outside the interval $[h,1-h]$. As we mentioned, the correct meaning of this quantity is that it constitutes the \textit{mid-point} of the confidence interval $[\hat{\vartheta}_t(S_t)-h,\hat{\vartheta}_t(S_t)+h]$ with the latter enjoying, for each fixed $t$, the largest possible coverage probability.
\end{remark}

Consider now the optimization in \eqref{eq:third_criterion} which will be performed in two steps: First we fix the stopping time $T$ and minimize $J(T,\hat{\theta}_T)$ with respect to $\hat{\theta}_T$; the resulting expression is then minimized, during the second step, over $T$ in order to obtain the optimum pair. We have the following lemma that addresses the first problem.

\begin{lemma}\label{lem:1}
Assume stopping time $T$ is fixed and satisfies $T\leq N$, where $N>0$ is some deterministic integer. Then,
\begin{equation}
J(T,\hat{\theta}_T)=c\Exp[T] + \Pro(|\hat{\theta}_T-\theta|>h)\geq
\Exp[cT+\C_{T}]=\J(T),
\label{eq:step1}
\end{equation}
with equality when we apply the corresponding Bayesian estimator $\hat{\theta}_T=\hat{\vartheta}_T$ at the time of stopping.
\end{lemma}
\begin{proof}
The proof is straightforward and presented in the Appendix.
\end{proof}
A side-product of Lemma\,\ref{lem:1}, as it can be verified from the corresponding proof in the Appendix, is the fact that the Bayesian estimator is not only optimum for fixed sample size, but it retains its optimality property when the sample size is controlled by any stopping time $T$ adapted to the observations.

Using \eqref{eq:step1} from Lemma\,\ref{lem:1}, we are now left with the optimization of the stopping time $T$. Assuming that $N$ is an integer which is sufficiently large, we consider the following optimization over stopping times that are bounded by $N$
\begin{equation}
\inf_{0\leq T\leq N}\J(T)=\inf_{0\leq T\leq N}\Exp[cT+\C_T].
\label{eq:step2}
\end{equation}
This is a classical \textit{finite horizon} optimal stopping problem with cost per sample equal to $c$ and cost for stopping at $t$ equal to $\C_t$. Of course, it is only natural to wonder why we limited our analysis to finite horizons instead of considering the more classical \textit{infinite horizon} version. As we will see in the sequel, for the most common prior we will be able to demonstrate that the infinite horizon assumption is completely unnecessary. Indeed, the optimum stopping time will turn out to be bounded by a deterministic quantity, suggesting that by limiting ourselves to a (sufficiently large) finite horizon, we do not suffer any performance loss. 

In order to solve the optimization problem defined in \eqref{eq:step2}, we follow the classical optimal stopping theory \cite{Shiryaev}. For $t=0,1,\ldots,N$ define the sequence of optimal \textit{average residual costs}
\begin{equation}
\V_t=\inf_{t\leq T\leq N}\Exp[c(T-t)+\C_{T}|\F_t],
\end{equation}
then we have
\begin{equation}
\V_t=\min\{\C_t,c+\Exp[\V_{t+1}|\F_t]\},~t=N,\ldots,1,0,
\label{eq:backwards}
\end{equation}
with the backward recursion initialized with $\V_{N+1}=1$. Regarding this last selection, it produces $\V_N=\C_N$ since the latter is a probability. In fact, this is exactly what the optimum residual cost at $N$ must be, because if we have not stopped before $N$, then we necessarily stop at $N$ and this produces cost $\C_N$ (simply the cost of stopping at $N$). The total optimum cost is expressed through $\V_0$, namely $\V_0=\inf_{0\leq T\leq N}\J(T)$. The next lemma specifies in more detail the recursion in \eqref{eq:backwards}.

\begin{lemma}\label{lem:2}
Consider the recursion in \eqref{eq:backwards} then, the optimal residual cost $\V_t, t=N,\ldots,0$ is a function $\V_t(S_t)$ of $S_t$ and therefore $\F_t$-measurable. Furthermore, \eqref{eq:backwards} can be written as
\begin{equation}
\V_t(S_t)=\min\{\C_t(S_t),c+\tilde{\V}_t(S_t)\},~t=N,\ldots,0,
\label{eq:backwards2}
\end{equation}
where $\tilde{\V}_t(S_t)$ expresses the optimum average residual cost to continue, satisfying
\begin{align}
\tilde{\V}_t(S_t)&=g_{t+1}(S_t)\V_{t+1}(S_t+1)+\big(1-g_{t+1}(S_t)\big)\V_{t+1}(S_t),\label{eq:barv}\\
g_{t+1}(S_t)&=\Pro(X_{t+1}=1|\F_t)=\frac{\int_0^1\theta^{S_t+1}(1-\theta)^{t-S_t}\pi(\theta)\,d\theta}{\int_0^1\theta^{S_t}(1-\theta)^{t-S_t}\pi(\theta)\,d\theta}.
\label{eq:gt}
\end{align}
Finally, if the prior $\pi(\theta)$ is symmetric around $\frac{1}{2}$ then the functions $\C_t(S_t),\V_t(S_t),\tilde{\V}_t(S_t)$ are symmetric with respect to $S_t$ around the value $\frac{t}{2}$.
\end{lemma}
\begin{proof} 
The validity of this lemma is straightforward and can be easily established using induction. We therefore give no further details.
\end{proof}
Once the sequence of optimal residual costs has been obtained through the solution of \eqref{eq:backwards2}, it is then immediate to define the optimum stopping time $\T$ that solves the minimization problem in \eqref{eq:step2}. Again, optimal stopping theory \cite{Shiryaev} suggests that
\begin{equation}
\T=\inf\{0\leq t\leq N: \V_t(S_t)=\C_t(S_t)\}
=\inf\{0\leq t\leq N: \C_t(S_t)\leq c+\tilde{\V}_t(S_t)\}.
\label{eq:opt-stop}
\end{equation}
In other words, when the optimum residual cost $\V_t(S_t)$ matches, for the first time, the cost for stopping $\C_t(S_t)$ or, equivalently, the cost of stopping is smaller than the residual cost of continuing, 
this is when we stop. Since the functions involved depend on $S_t$, this quantity can serve as our test statistic and we can express the stopping rule in \eqref{eq:opt-stop} in terms of $S_t$. Specifically, for each time $t$, we can find the sampling region $\Omega_t=\{0\leq S_t\leq t: \V_t(S_t)<\C_t(S_t)\}=\{0\leq S_t\leq t: c+\tilde{\V}_t(S_t)<\C_t(S_t)\}$ with $\Omega_N=\varnothing$, and we can equivalently define the stopping time as $\T=\inf\{0\leq t\leq N: S_t\not\in\Omega_t\}$. 

\subsection{The Constrained Problem}
Let us now turn to the constrained problem in \eqref{eq:second_criterion} which we can solve with the results we have so far. We will show that \eqref{eq:second_criterion} can be recovered as an instance of the unconstrained version \eqref{eq:third_criterion} corresponding to a special selection of the Lagrange multiplier $c$. Our result is summarized in the following theorem.

\begin{theorem}\label{th:1} 
For the solution of \eqref{eq:second_criterion} we distinguish two cases:

\noindent i)~If $\alpha\geq\C_0=\Pro(|\hat{\vartheta}_0-\theta|>h)$, with $\hat{\vartheta}_0={\rm arg}\inf_{\hat{\theta}_0}\Pro(|\hat{\theta}_0-\theta|> h)$, then the optimum is to stop without taking any samples, i.e. $\T=0$ and use as mid-point of the optimum confidence interval the value $\hat{\vartheta}_0$ which is based only on the prior $\pi(\theta)$.

\noindent ii)~If $\Pro(|\hat{\vartheta}_0-\theta|> h)>\alpha$, then for any horizon $N\geq N_{\alpha}$ where $N_\alpha$ satisfies $\Pro(|\hat{\vartheta}_{N_\alpha}-\theta|> h)<\alpha$, there exists Lagrange multiplier $c_*$ such that the solution of \eqref{eq:third_criterion} is also the solution to \eqref{eq:second_criterion} that can involve a possible randomization before taking any samples.
\end{theorem}
\begin{proof}
The proof of this theorem is presented in the Appendix.
\end{proof}

\section{Properties of the Optimum Solution}

If we fix the value $N$ of the horizon and the cost per sample $c$, we can then compute the mid-points $\{\{\hat{\vartheta}_t(S_t)\}_{S_t=0}^t\}_{t=0}^N$ of the confidence intervals from \eqref{eq:1}. Assuming that $\pi(\theta)$ is continuous, candidates for $\hat{\vartheta}_t(S_t)$ can be obtained from the solution of the following equation which we obtain by differentiating \eqref{eq:1} with respect to $\hth_t$
\begin{equation}
(\hat{\theta}_t+h)^{S_t}(1-\hat{\theta}_t-h)^{t-S_t}\pi(\hat{\theta}_t+h)-(\hat{\theta}_t-h)^{S_t}(1-\hat{\theta}_t+h)^{t-S_t}\pi(\hat{\theta}_t-h)=0.
\label{eq:root}
\end{equation}
The previous equation has clearly a solution in the interval $[h,1-h]$ when $0<S_t<t$ with the corresponding value providing a (local) extremum for the coverage probability. To these candidate mid-points we must include the two end points $h,1-h$ since the global maximum can occur at the two ends as well. Therefore, we need to examine which of these cases provides the best coverage probability and select the corresponding value as our optimum mid-point $\hat{\vartheta}_t(S_t)$. When $S_t=0,t$ it is possible \eqref{eq:root} not to have any solution in $[h,1-h]$. In this case, $\hat{\vartheta}_t(0)$ and $\hat{\vartheta}_t(t)$ are equal to one of the two end values $h$ or $1-h$. Having identified the optimum mid-points $\{\{\hat{\vartheta}_t(S_t)\}_{S_t=0}^t\}_{t=0}^N$, we apply \eqref{eq:2} to compute the corresponding optimum complementary conditional coverage probabilities $\{\{\C_t(S_t)\}_{S_t=0}^t\}_{t=0}^N$.

The next step consists in computing $\{\{g_{t+1}(S_t)\}_{S_t=0}^t\}_{t=0}^N$ for $t=0,\ldots,N$ and $S_t=0,\ldots,t$ with numerical integration. Once we have available 
$\{\{\C_t(S_t)\}_{S_t=0}^t\}_{t=0}^N$ and $\{\{g_{t+1}(S_t)\}_{S_t=0}^t\}_{t=0}^N$, we can then use them in the backward recursion \eqref{eq:backwards2} to find the sequence $\{\{\tilde{\V}_t(S_t)\}_{S_t=0}^t\}_{t=0}^N$ and the optimum residual cost sequence $\{\{\V_t(S_t)\}_{S_t=0}^t\}_{t=0}^N$.
To identify the stopping rule, according to \eqref{eq:opt-stop} we must compare the two sequences $\{\{\C_t(S_t)\}_{S_t=0}^t\}_{t=0}^N$, $\{\{\V_t(S_t)\}_{S_t=0}^t\}_{t=0}^N$ element-by-element. At coordinates $(t,S_t)$ where the sequences differ, we decide to continue sampling; whereas if they are equal, we decide to stop. This generates the sequence of sampling regions $\{\Omega_t\}_{t=0}^N$. Equivalently, we can compare $\{\{\C_t(S_t)\}_{S_t=0}^t\}_{t=0}^N$ with $\{\{c+\tilde{\V}_t(S_t)\}_{S_t=0}^t\}_{t=0}^N$, and wherever the first is no larger than the second, we stop, while we continue sampling in the opposite case.

We now present a conjecture that contains two significant claims for the optimum stopping time for the problem in \eqref{eq:third_criterion} which we believe are valid for \textit{any} prior $\pi(\theta)$. We were able to provide a proof for the first claim (Lemma\,\ref{lem:2a}) for a rich class of priors, and prove both claims (Theorem\,\ref{th:2}) providing also quantitative information when the prior is the Beta density. Regarding the latter case we should note that the Beta density is among the most popular priors for the problem we are considering in this work.

\begin{conjecture*}
For any prior $\pi(\theta)$ and sufficiently large horizon $N$ the optimum stopping time $\T$ of the unconstrained problem in \eqref{eq:third_criterion} enjoys the following two properties:\\
i).~There exists constant $t_{\rm up}$ depending only on $c$ and not on $N$ such that $\T\leq t_{\rm up}$.\\
ii).~For sufficiently small $c$ there exists constant $t_{\rm lo}\geq1$ depending only on $c$ and not on $N$ such that $t_{\rm lo}\leq\T$.
\end{conjecture*}

Below we present a general proof of property i) of the Conjecture under the following additional assumption: Define the maximal conditional variance
\begin{equation}
\sigma_t^2 =\max_{0\leq S_t\leq t}\Exp\big[\big(\theta-\Exp[\theta|S_t]\big)^2|S_t\big]= \max_{0\leq S_t\leq t} \int_0^1 (\theta - \Exp[\theta|S_t])^2 \pi_t(\theta|S_t) d \theta,
 \label{eq:Vstar}
\end{equation}
where $\pi_t(\theta|S_t)$ is the posterior pdf defined in (\ref{eq:posterior}) and assume that
$\sigma_t  \to 0$ as $t \to \infty$. This forces the conditional variance to converge to 0 \textit{uniformly} in $S_t$. It also implies that the posterior distribution $\pi_t(\theta|S_t)$ converges, uniformly, to a degenerate measure at a single point (often the true $\theta$) as $t \to \infty$. This is clearly related to the {\it consistency} concept of posterior distributions in Bayesian statistics and is often considered a valid assumption (see \cite{choi}).

\begin{lemma}\label{lem:2a} Let $\sigma_t$ be defined as in \eqref{eq:Vstar} with $\lim_{t \to \infty} \sigma_t = 0$. Then for sufficiently large horizon there exists constant $t_{\rm up}$ depending only on $c$ such that $\T\leq t_{\rm up}$, i.e. property i) in the Conjecture is true.
\end{lemma}

\begin{proof} The proof is a simple application of the Chebyshev inequality in combination with \eqref{eq:Vstar}. Indeed we observe that
\begin{equation}
\C_t(S_t) = \inf_{\hat{\theta}_t}\Pro(|\theta-\hat{\theta}_t|>h|\F_t) \leq \Pro\big(|\theta-\Exp[\theta|\F_t]|>h|\F_t\big)\leq \frac{1}{h^2}\Exp\big[\big(\theta-\Exp[\theta|S_t]\big)^2|S_t\big]
\leq \frac{\sigma_t^2}{h^2}.
\end{equation}
Since $\sigma_t \to 0$ as $t \to \infty,$ there exists $N$ such that $\C_N\leq \frac{\sigma^2_N}{h^2}\leq c$ and, therefore, from \eqref{eq:backwards2} we conclude that $\C_N\leq c+\tilde{\V}_N$, which suggests that we will necessarily stop at $N$ for any value of $S_N$. Quantity $t_{\rm up}$ is the smallest $N$ for which this is true.
\end{proof}
%
%

\begin{remark}
The assumption $\lim_{t \to \infty} \sigma_t = 0$ does not hold for all prior distribution. A counterexample where it fails is when the prior is a two-point probability mass function, say $\Pro(\theta = 0.4) = \Pro(\theta = 0.6) = 0.5$. However, even for this case the Conjecture might still be valid since the requirement $\C_N(S_N)\leq\frac{\sigma_N^2}{h^2}<c$ used in our proof, is only sufficient for the validity of our claim.
\end{remark}

An interesting example where the assumption holds is when the prior is the Beta density $\pi(\theta) =\text{Beta}(\theta,p,q)$, where
\begin{equation}
\text{Beta}(\theta,p,q)=\frac{\theta^{p-1}(1-\theta)^{q-1}}{\int_0^1\theta^{p-1}(1-\theta)^{q-1}\,d\theta},~~p,q>0.
\label{eq:beta}
\end{equation}
To see this, we note that the posterior pdf is of the same type, namely $\pi(\theta|S_t) =\text{Beta}(\theta,p+S_t,t-S_t+q)$, and thus the maximal conditional variance in \eqref{eq:Vstar} becomes
\begin{equation}
\sigma_t^2 = \max_{0 \le S_t \le t} \frac{(p+S_t)(t-S_t+q)}{(t+p+q)^2 (t+p+q+1)} \le \frac{1}{4(t+p+q+1)},
\end{equation}
where the equality is attainable when $S_t = \frac{t+q-p}{2}$ is an integer.  Clearly, for fixed $p,q > 0$ we have $\sigma_t \to 0$ as $t \to \infty$, and thus the assumption of Lemma\,\ref{lem:2a} holds. Moreover, by the proof of Lemma\,\ref{lem:2a}, the optimum stopping time satisfies  $\T\leq \max\{0, \frac{1}{4 h^2 c} - p-q-1\}$ for all $c > 0.$  This bound is of the order of $c^{-1}$. In Theorem\,\ref{th:2}, Section~\ref{ssec:3B}, by applying a more advanced analysis, we will be able to improve it and provide an alternative estimate which is of the order of $|\log(c)|$ for the case of the symmetric prior $p=q$.

\begin{remark}
Property i) of the Conjecture suggests that the number of samples, under the optimum scheme, will never exceed the value $t_{\rm up}$ even if we allow the horizon to grow without limit. This interesting and uncommon characteristic was also observed in \cite{c77} but with cost function a variance of the classical mean square error. However, what is more intriguing in our conjecture is property ii), namely that we need first to accumulate a sufficient volume of information before we start asking ourselves whether we should stop sampling or not. This is an extremely uncommon feature and, to our knowledge, has never been reported before in Sequential Analysis as a property of optimum schemes. As we claim in our conjecture, we believe that both properties are valid for any prior $\pi(\theta)$. Fortunately, as we mentioned before, this double claim is not without solid evidence. Indeed with Theorem\,\ref{th:2}, we demonstrate its validity when the prior is the symmetric Beta density.
\end{remark}

\subsection{Performance Evaluation}

What we presented so far allows for the determination of the stopping rule of the proposed scheme. We would like now to compute its performance but also the performance of any stopping time which uses $S_t$ as its test statistic and is defined in terms of a sequence of sampling regions $\{\Omega_t\}$ in terms of $\{S_t\}$. In particular, we are interested in computing $\Exp[T|\theta], \Exp[T], \Pro(|\hat{\theta}_{T}-\theta|\leq h|\theta)$ and $\Pro(|\hat{\theta}_{T}-\theta|\leq h)$. Of course, we could obtain these quantities using Monte-Carlo simulations, but it is also possible to determine them numerically. The following lemma provides the necessary formulas.

\begin{lemma}\label{lem:3} Let the stopping time $T$ be bounded by $N$ having as test statistic the process $\{S_t\}$. Assume for each $t$ that $\Omega_t$ denotes the sampling region. Suppose also that for the combination $(t,S_t)$ the scheme provides the mid-point estimate $\hat{\theta}_t(S_t)$ and the corresponding conditional complementary coverage probability $C_t(S_t)=\Pro(|\hat{\theta}_t(S_t)-\theta|> h|\F_t)$. For $t=N-1,\ldots,0$, we then define the following backward recursions that must be applied for $S_t=0,1,\ldots,t$
\begin{align}
U_t(S_t)&=1+\theta \ind{S_t+1\in\Omega_{t+1}}U_{t+1}(S_t+1)+(1-\theta)\ind{S_t\in\Omega_{t+1}}U_{t+1}(S_t),\label{eq:lem3-1}\\
\bar{U}_t(S_t)&=1+g_{t+1}(S_t)\ind{S_t+1\in\Omega_{t+1}}\bar{U}_{t+1}(S_t+1)+\big(1-g_{t+1}(S_t)\big)\ind{S_t\in\Omega_{t+1}}\bar{U}_{t+1}(S_t),\label{eq:lem3-2}\\
W_t(S_t)&=\ind{|\hth_t-\theta|> h}\ind{S_t\not\in\Omega_t}+\big\{\theta W_{t+1}(S_t+1)+(1-\theta)W_{t+1}(S_t)\big\}\ind{S_t\in\Omega_t},\label{eq:lem3-3}\\
\bar{W}_t(S_t)&=C_t(S_t)\ind{S_t\not\in\Omega_t}+\big\{g_{t+1}(S_t) \bar{W}_{t+1}(S_t+1)+\big(1-g_{t+1}(S_t)\big)\bar{W}_{t+1}(S_t)\big\}\ind{S_t\in\Omega_t},\label{eq:lem3-4}
\end{align}
where $g_{t+1}(S_t)$ is defined in \eqref{eq:gt} and the four recursions are initialized with $U_N(S_N)=\bar{U}_N(S_N)=0, W_N(S_N)=\ind{|\hth_N-\theta|> h},\bar{W}_N(S_N)=C_N(S_N),\Omega_N=\varnothing$.
Then, $\Exp[T|\theta]=U_0(S_0),\Exp[T]=\bar{U}_0(S_0),\Pro(|\hat{\theta}_T-\theta|>h|\theta)=W_0(S_0)$ and 
$\Pro(|\hat{\theta}_T-\theta|>h)=\bar{W}_0(S_0)$.
\end{lemma}
\begin{proof}
The validity of these expressions is established in the Appendix.
\end{proof}
The applicability of Lemma\,\ref{lem:3} is clearly not limited to the proposed scheme but can be used to compute the performance of the fixed-sample-size and of other sequential alternatives that we intend to compare against the method we have developed.

\subsection{Beta Density as Prior}\label{ssec:3B}
Let us now find the particular form of our scheme when we adopt as our prior the Beta density $\pi(\theta) =\text{Beta}(\theta,a,a)$, where $\text{Beta}(\theta,p,q)$ is defined in \eqref{eq:beta}.
We observe that the selection $a=1$ in the prior corresponds to the uniform density in $[0,1]$. It is now straightforward to verify that the posterior pdf accepts a similar form, namely
\begin{equation}
\pi(\theta|S_t) =\text{Beta}(\theta,a+S_t,a+t-S_t),
\end{equation}
while the conditional complementary coverage probability at time $t$ becomes
\begin{equation}
\Pro(|\hat{\theta}_t-\theta|>h|\F_t)=1-I_{\min\{1,\hth_t+h\}}(a+S_t,a+t-S_t)+I_{\max(0,\hth_t-h)}(a+S_t,a+t-S_t),
\label{eq:25}
\end{equation}
where $I_x(p,q)$ is the incomplete Beta function (see \cite[Page 944]{Abramowitz}) which is the cdf of $\text{Beta}(\theta,p,q)$.

The Bayes estimator, according to \eqref{eq:root}, can be found as the solution of the equation
$$
\hat{\vartheta}_t=\text{arg}\left\{\hat{\theta}_t:\left(\frac{\hat{\theta}_t-h}{\hat{\theta}_t+h}\right)^{a+S_t-1}=\left(\frac{1-h-\hat{\theta}_t}{1+h-\hat{\theta}_t}\right)^{a+t-S_t-1}\right\} 
$$
corresponding to the root in the interval $[h,1-h]$. Such root always exists except when $a=1$ and $S_t=0$ or $t$. For these cases, $\hat{\vartheta}_t$ is equal to $h$ or $1-h$, depending on which value provides a larger conditional coverage probability. The resulting optimum conditional complementary coverage probability becomes  
\begin{equation}
\C_t(S_t)=1-I_{\min\{1,\hvth_t+h\}}(a+S_t,a+t-S_t)+I_{\max(0,\hvth_t-h)}(a+S_t,a+t-S_t).
\label{eq:15}
\end{equation}
Finally, as indicated in \eqref{eq:barv} and \eqref{eq:gt}, we need to find the probability $g_{t+1}(S_t)$, for which we have the following simple formula
\begin{equation}
g_{t+1}(S_t)=\Pro(X_{t+1}=1 |\F_t)=\frac{\Gamma(S_t+a+1)\Gamma(t+2a)}{\Gamma(S_t+a)\Gamma(t+2a+1)}=\frac{S_t+a}{t+2a}.
\end{equation}
We can now compute the sequences $\{\{\V_t(S_t)\}_{S_t=0}^t\}_{t=0}^N,\{\{\tilde{\V}_t(S_t)\}_{S_t=0}^t\}_{t=0}^N$ as explained in \eqref{eq:backwards2} and compare, element-by-element, $\{\{\C_t(S_t)\}_{S_t=0}^t\}_{t=0}^N$ with $\{\{\V_t(S_t)\}_{S_t=0}^t\}_{t=0}^N$ or $\{\{\C_t(S_t)\}_{S_t=0}^t\}_{t=0}^N$ with $\{\{c+\tilde{\V}_t(S_t)\}_{S_t=0}^t\}_{t=0}^N$ to identify the sampling and stopping regions.

For the particular prior adopted in \eqref{eq:beta}, as we mentioned before, the resulting optimum stopping time $\T$ enjoys the unique properties claimed in the Conjecture. The next theorem provides the necessary evidence.

\begin{theorem}\label{th:2}
The Conjecture is true when the prior is the Beta density $\pi(\theta) =\text{Beta}(\theta,a,a)$ with the optimum stopping time $\T$ satisfying $C_0 |\log(c)| \leq \T \leq C_1 |\log(c)|$ for constants $C_0 < C_1$ that depend only on $a$ and $h$.
\end{theorem}
\begin{proof}
The proof is very technical and detailed in the Appendix. Unfortunately, the analytical techniques developed for the specific prior are not directly extendable to the general case.
\end{proof}
Perhaps, it is worth mentioning the fact that from the proof of Theorem\,\ref{th:2}, we conclude that the two estimates for $t_{\rm up}$ and $t_{\rm lo}$ in \eqref{eq:N},\eqref{eq:nu} grow linearly in $|\log(c)|$ having drastically different multiplicative coefficients ($C_0$ of the order of $\frac{1}{2h^2}$ versus $C_1$ of the order of $\frac{1}{|\log(0.5-h)|}$) and different offsets. 

\begin{figure}[h]
\centering
  \includegraphics[width=3.2in]{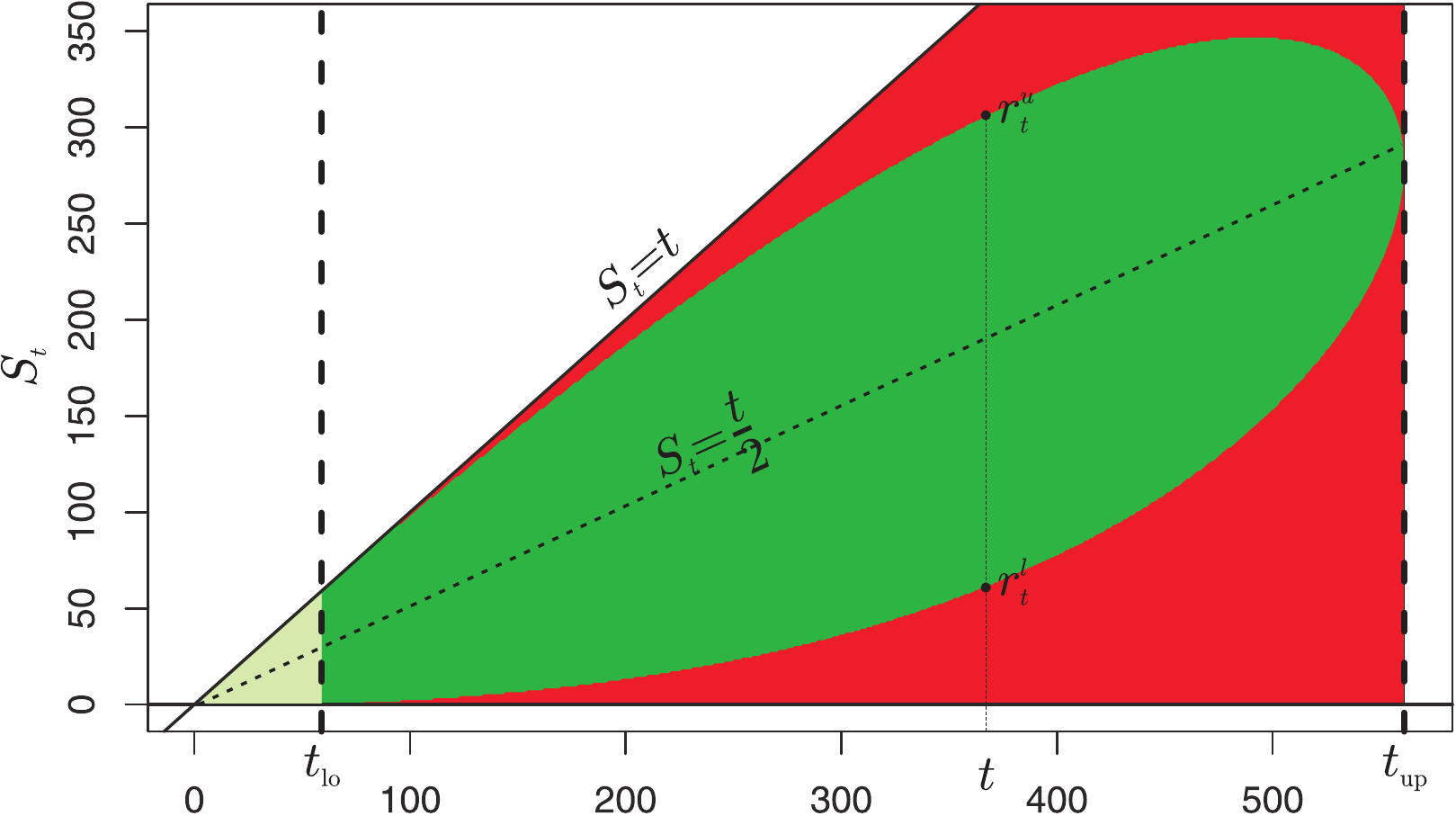}%
\caption{Sampling (green) and stopping (red) regions for $a=1, h = 0.05$ and $c = 0.0001$. Upper and lower bounds for optimum stopping time: $t_{\rm lo}=59$ and $t_{\rm up}=561$. No possibility of stopping (light green).}
\label{fig:Fig1}
\end{figure}
As an illustration for these properties we consider $a=1, h=0.05$, and $c=0.0001$. Fig.\,\ref{fig:Fig1} depicts the sampling (green) and the stopping (red) region in terms of the test statistic $S_t$. Both regions are clearly limited between the lines $S_t=t$ and $S_t=0$. Even though we have marked a whole region in red, only the points that are next to the green region are actually accessible because $S_t$ can increase at most by one unit as we go from $t$ to $t+1$. We can also see the two bounds $t_{\rm up}=561$ and $t_{\rm lo}=59$ for $\T$. For $t\leq t_{\rm lo}$ the light green region covers all points $0\leq S_t\leq t$, thus identifying the time instances we can never stop. Also, 
we note that once we pass $t_{\rm up}$ we are in the stopping region suggesting that we must necessarily stop at $t_{\rm up}$. For each $t_{\rm lo}\leq t\leq t_{\rm up}$ the stopping region has an upper $r_t^u$ and a 
lower $r_t^l$ threshold and, as long as $S_t$ is between these two limits, we need to sample. Since the prior distribution is symmetric with respect to $1/2$, then, according to Theorem\,\ref{th:1}, the sampling region is symmetric around $t/2$, implying that $r_t^u+r_t^l=t$.

\begin{figure}[b]
\centering
\includegraphics[width=3.2in]{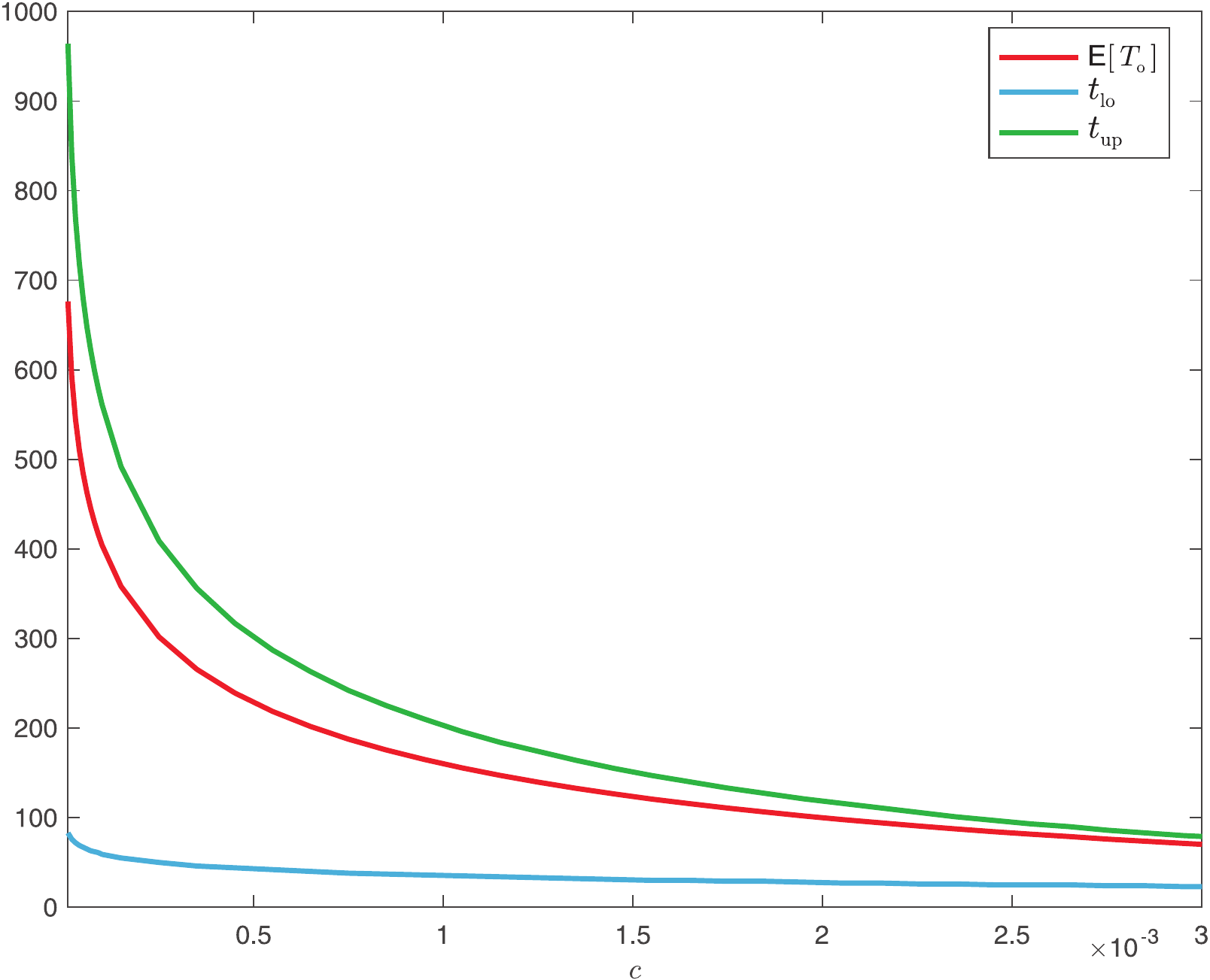}
\caption{Average sample size (red), lower $t_{\rm lo}$ (blue) and upper $t_{\rm up}$ limit (green), as functions of $c$ for optimum stopping time $\T$ when $a=1$ and $h=0.05$.}
\label{fig:bounds}
\end{figure}
In Fig.\,\ref{fig:bounds}, after using \eqref{eq:lem3-2}, we plot the average sample size and the two limits $t_{\rm lo},t_{\rm up}$  of $\T$ as functions of $c$ for $a=1$ and $h=0.05$. We can see that the lower limit $t_{\rm lo}$ is significantly smaller than the resulting average, suggesting that the optimum scheme very quickly regards the accumulated information as capable of providing reliable interval estimates and therefore starts the process of questioning whether to stop or continue sampling.

\section{Comparisons}
Let us now compare our scheme with the optimal fixed-sample-size (FSS) and two sequential methods: The first was proposed by Frey in \cite{f10} and the second, the \textit{Conditional Method}, was proposed in our earlier work in \cite{gym17}. Frey's method uses a modified Wald-type  sequential confidence interval based on the stopping time 
\begin{eqnarray} 
T_{\rm F} = \inf\left\{t \ge 0:   \frac{\tilde{\theta}_{t, k} (1 - \tilde{\theta}_{t, k})}{t}\leq \left(\frac{h}{z_{\frac{\gamma}{2}}}\right)^2  \right\},
\label{eqn_freystop}
\end{eqnarray}
where $\tilde{\theta}_{t, k} = \frac{S_t+k}{t+2k}$, $k>0$ is a pre-specified constant and $\gamma = \gamma(k, h, \alpha)$ is chosen so that the confidence interval $[\hat{\theta}_{T_{\rm F}}-h, \hat{\theta}_{T_{\rm F}}+h]$, with $\hat{\theta}_{t}=\frac{S_t}{t}$, has a confidence level of at least $1-\alpha$. Table\,\ref{tab:1} provides the values of $k$ and $\gamma$ recommended in \cite{f10} for best results.
\begin{table}[h]
\caption{Choices of $k$ and $\gamma$ for $90\%$, $95\%$, and $99\%$ confidence intervals
of fixed half-width $h$ in \cite{f10}.}
\centering{}%
\begin{tabular}{|c|c|c|c|c|c|c|}
\hline
 & \multicolumn{2}{c|}{ $90\%$} & \multicolumn{2}{c|}{ $95\%$} & \multicolumn{2}{c|}{ $99\%$}\tabularnewline
\hline
\hline
$h$  & $k$  & $\gamma$  & $k$  & $\gamma$  & $k$  & $\gamma$\tabularnewline
\hline
$0.10$  & $4$  & $0.0754$  & $4$  & $0.0356$  & $6$  & $0.0068$\tabularnewline
\hline
$0.05$  & $4$  & $0.0859$  & $6$  & $0.0433$  & $8$  & $0.0083$\tabularnewline
\hline
$0.01$  & $8$  & $0.0972$  & $10$  & $0.0487$  & $14$  & $0.0097$\tabularnewline
\hline
\end{tabular}
\label{tab:1}
\end{table}
From \eqref{eqn_freystop} and using the fact that $x(1-x)\leq\frac{1}{4}$ we conclude that the corresponding stopping time satisfies $T_{\rm F}\leq\lceil\frac{z_{\frac{\gamma}{2}}^2}{4h^2}\rceil=N$.
Regarding the finite-sample-size method, it uses the optimum Bayes estimator $\hvth_t$, obtained in \eqref{eq:1} and the number of samples $t$ is selected to meet the desired coverage probability. Finally, for the conditional method in \cite{gym17}, we should point out that it is a general sequential parameter estimation technique based on conditional costs which is not limited to binomial proportions. For the problem of interest, we have $T_{\rm C}=\inf\{t\geq0: \C_t\leq \beta\}$ and $\hat{\theta}_{T_{\rm C}}=\hat{\vartheta}_{T_{\rm C}}$, where $\hat{\vartheta}_t,\C_t$ are the Bayes estimator and the corresponding optimum conditional complementary coverage probability defined in \eqref{eq:1},\eqref{eq:2}. Threshold $\beta$ is selected to guarantee that the resulting coverage probability is $1-\alpha$. For $\C_t$ we have from the proof of Theorem\,\ref{th:1}, eq. \eqref{eq:A1}, that $\C_t\leq 2e^{-2h^2(t+2a+1)}$, consequently $T_{\rm C}\leq \lceil\max\{\frac{|\log(\frac{\beta}{2})|}{2h^2}-2a-1,0\}\rceil=N$. In other words, all four schemes satisfy the assumption of Lemma\,\ref{lem:3} of bounded stopping time, therefore the corresponding performance can be computed numerically by applying the recursions of the lemma without the need to perform Monte-Carlo simulations.

\begin{figure}
\centering
\includegraphics[width=3.2in]{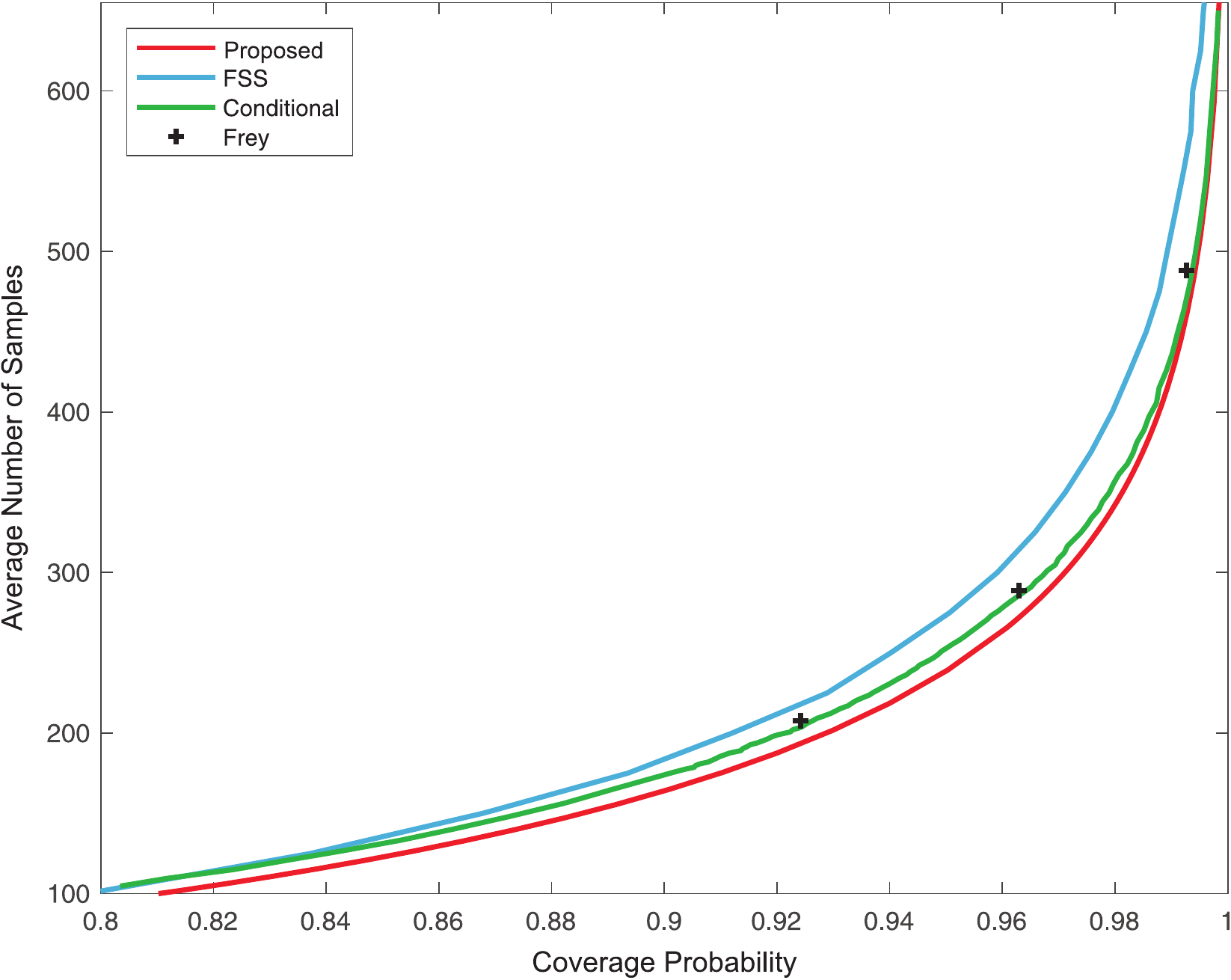}
\caption{Average samples size versus coverage probability for proposed (red), Frey (black $+$), fixed-sample-size (blue) and conditional (green), for $a=1$ and $h=0.05$.}
\label{fig:Fig3}
\end{figure}
For the competing methods using \eqref{eq:lem3-2},\eqref{eq:lem3-4}, we plot in Fig.\,\ref{fig:Fig3} the average number of samples $\Exp[T]$ versus the coverage probability $\Pro(|\hat{\theta}_T-\theta|\leq h)$ when $a=1$ and $h=0.05$. Note that we have three points for Frey's scheme because of the tuning parameters $k$ and $\gamma$ which are provided in Table\,I only for three confidence levels. As we can see, the proposed method outperforms the fixed-sample-size and both alternative sequential techniques. It is only at very high coverage probability levels that the difference between the three sequential schemes becomes less pronounced.

As we pointed out in \eqref{eq:second_criterion*}, Section\,II, there is practical interest in evaluating the performance for each individual $\theta$. Clearly in this case, the requirement is to be able to guarantee a minimal coverage probability \textit{for all} $\theta$. 
Again, we resort to Lemma\,\ref{lem:3} and use \eqref{eq:lem3-1},\eqref{eq:lem3-3} to evaluate the performance of the competing methods for each $\theta$. In Fig.\,\ref{fig:Fig4a}, we plot the coverage probability for each test versus $\theta$ and in Fig.\,\ref{fig:Fig4b}, the corresponding average sample size required to obtain this performance. Parameters were selected so as all competing schemes provide \textit{the same worst-case coverage probability} assuring a coverage of at least $0.95$ for all $\theta$. By observing the two figures, we can draw the following conclusions: The fixed-sample-size scheme can require up to almost eight times more samples compared to the proposed. Of course, one may argue that it produces higher coverage probability levels. Indeed this is true, but, unfortunately, this increased performance cannot be traded for a reduced sample size without compromising the worst-case level. Consequently, what we observe is in fact the best the fixed-sample-size method can offer. The conditional scheme, around $\theta=0.5$, requires up to 30\% more samples which, as in the case of fixed-sample-size, produce higher coverage probabilities. Again, it is impossible to sacrifice part of this increased performance to improve the corresponding sample size without degrading the worst-case coverage probability. Finally, we can see that the proposed and Frey's scheme require similar samples over most $\theta$. However, we observe that the proposed method has a coverage probability profile which is better than Frey's, since for most $\theta$ the corresponding probability is larger. Frey's scheme is slightly better only for $\theta$ close to 0 and 1. But even for these values of $\theta$ the proposed scheme requires almost 50\% less samples.
\begin{figure}[t!]
\centering
\subfloat[\label{fig:Fig4a}]{%
  \includegraphics[width=3.2in]{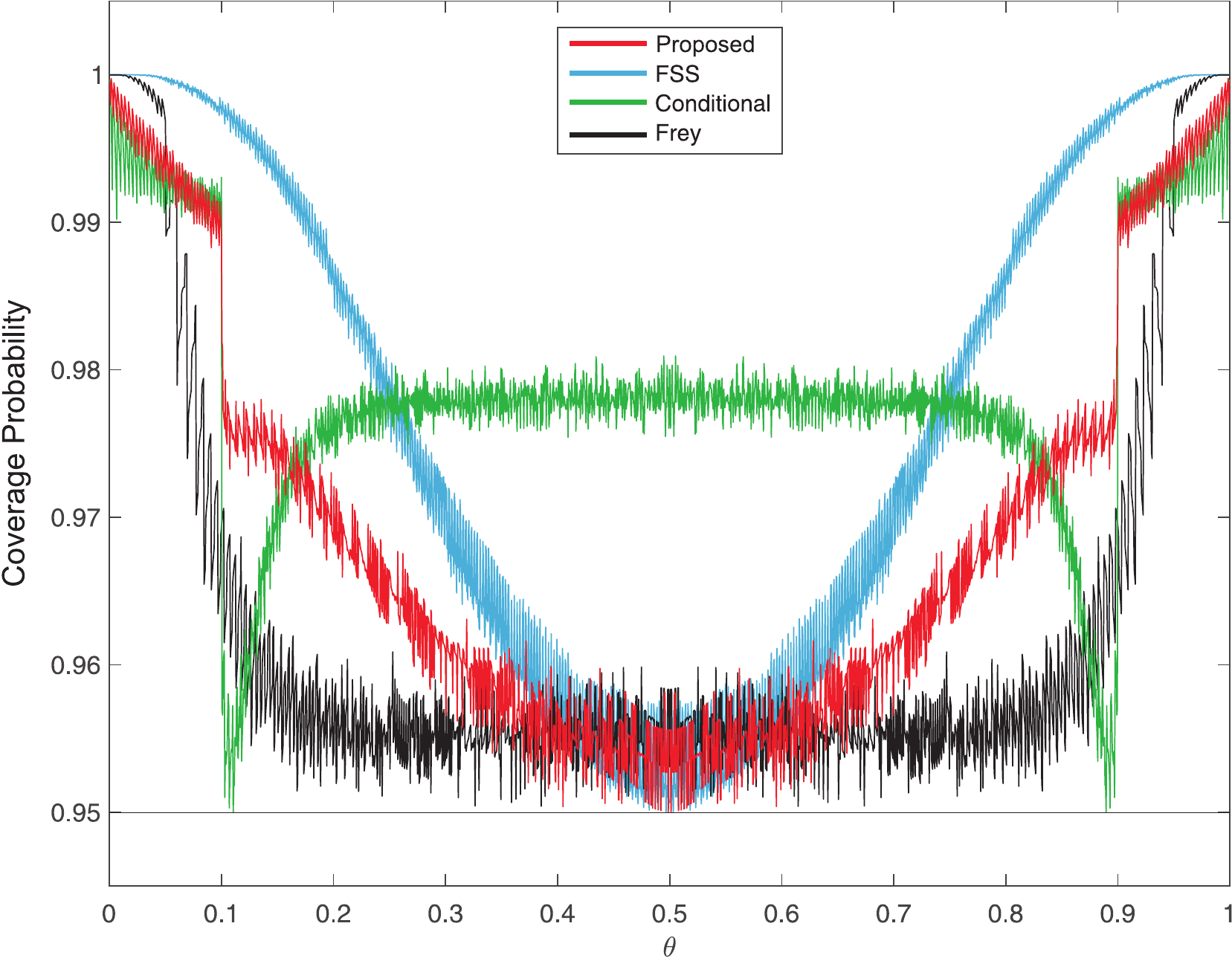}%
  }~
\subfloat[\label{fig:Fig4b}]{%
  \includegraphics[width=3.2in]{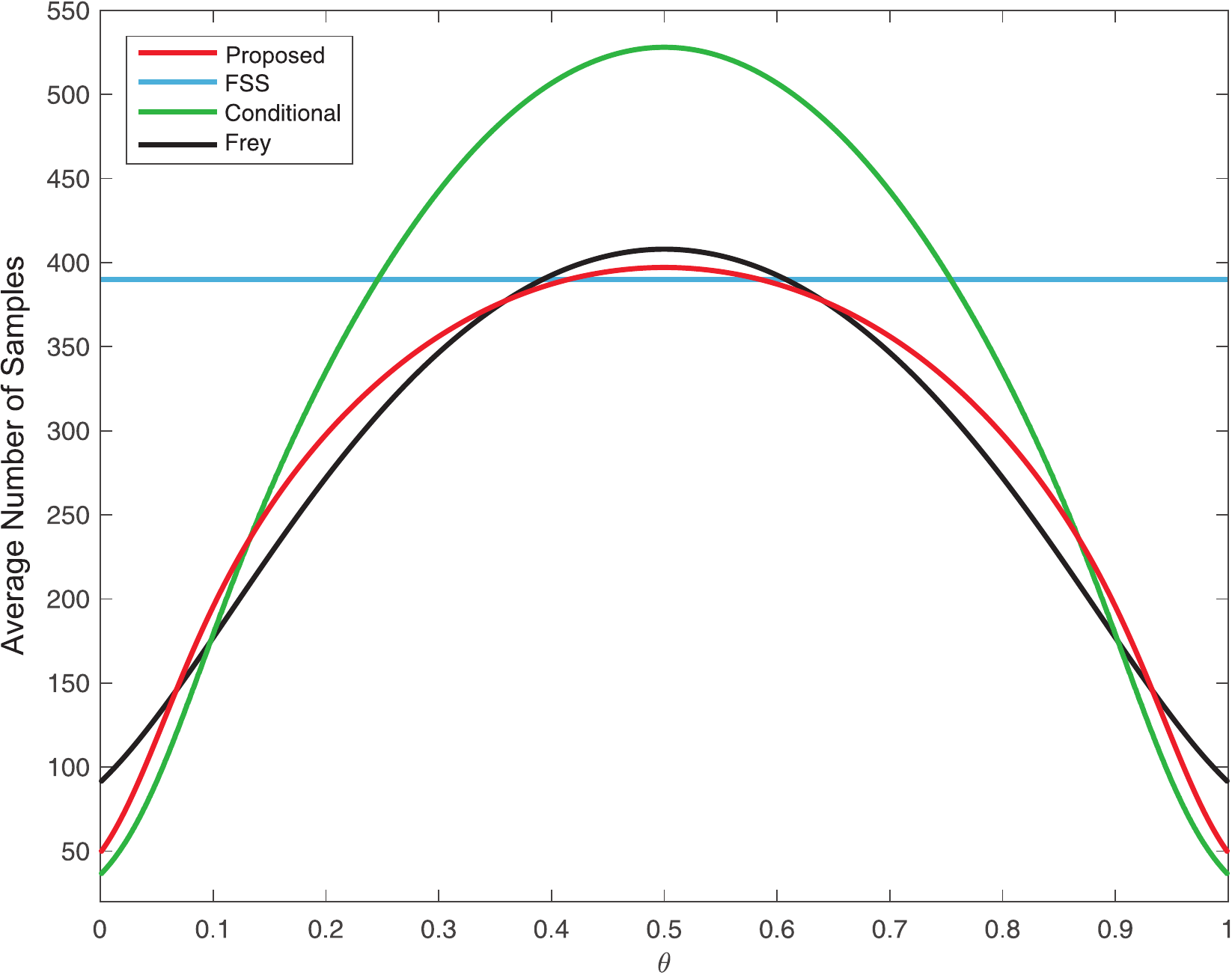}%
  }
\caption{Coverage probability (a) and Average sample size (b) as a function of proportion $\theta$ for proposed (red), Frey (black), fixed-sample-size (blue) and conditional (green) when $a=1$, $h=0.05$ and worst-case coverage probability $0.95$.}
\label{fig:Fig4}
\end{figure}

\section{Conclusions}
We proposed an optimal sequential scheme for obtaining confidence intervals for a binomial proportion under a well defined formulation. We proved that, for a particular prior (Beta density), our optimum stopping time enjoys certain uncommon properties not encountered in solutions of other classical optimal stopping problems. We also conjectured that these properties are present with any prior. Specifically, our claim is that our stopping time is always bounded from above and below, suggesting that we need to first accumulate a sufficient amount of information before we start applying our stopping rule, and that our stopping time will always terminate at a specific deterministic time even if we allow the time horizon to be infinite. Finally, our scheme was compared against the optimum fixed-sample-size procedure and against existing sequential alternatives. Numerical performance evaluations showed that the proposed method exhibits an overall improved performance profile compared to its rivals.


\section{Acknowledgments}
This work was supported by the US National Science Foundation under Grant CIF\,1513373 through Rutgers University and under Grant CMMI\,1362876 through Georgia Institute of Technology.

\appendix

\noindent\textit{Proof of Lemma\,\ref{lem:1}:} From \eqref{eq:step1} we can write
\begin{align}
J(T,\hat{\theta}_T)&=c\Exp[T] + \Pro(|\hat{\theta}_T-\theta|>h)=\sum_{t=0}^N\Exp\left[\left\{ct+\ind{|\hat{\theta}_t-\theta|>h}\right\}\ind{T=t}\right]\nonumber\\
&=\sum_{t=0}^N\Exp\left[\Exp\left[ct+\ind{|\hat{\theta}_t-\theta|>h}|\F_t\right]\ind{T=t}\right]
=\sum_{t=0}^N\Exp\left[\left\{ct+\Pro(|\hat{\theta}_t-\theta|>h|\F_t)\right\}\ind{T=t}\right]\label{eq:A1}\\
&\geq \sum_{t=0}^N\Exp\left[\left\{ct+\inf_{\hat{\theta}_t}\Pro(|\hat{\theta}_t-\theta|>h|\F_t)\right\}\ind{T=t}\right]
=\sum_{t=0}^N\Exp\left[\{ct+\C_t\}\ind{T=t}\right]\label{eq:A2}\\
&=\Exp[cT+\C_T].\nonumber
\end{align}
The first equality in \eqref{eq:A1} is true because $\ind{T=t}$ is $\F_t$-measurable, also
we have equality in \eqref{eq:A2} if we select $\hat{\theta}_t=\hat{\vartheta}_t$ when $\{T=t\}$. We observe that changing the order of summation and expectation presents absolutely no complication because the stopping time is bounded by the deterministic quantity $N$. \qed
\vskip0.2cm
\noindent\textit{Proof of Theorem\,\ref{th:1}:} If $\alpha\geq\Pro(|\hat{\vartheta}_0-\theta|> h)$ then stopping at $\T=0$ corresponds to the smallest possible (average) number of samples while, at the same time, we satisfy the coverage probability constraint.

To prove ii) we first show that there exists $N_\alpha$ such that $\Pro(|\hat{\vartheta}_{N_\alpha}-\theta|> h)<\alpha$.
Note that 
\begin{multline}
\Pro(|\hat{\vartheta}_t-\theta|> h)\leq \Pro\left(\left|\frac{S_t}{t}-\theta\right|> h\right)\leq\frac{1}{h^2}\Exp\left[\left(\frac{S_t}{t}-\theta\right)^2\right]\\
=\frac{1}{h^2}\Exp\left[\Exp\left[\left(\frac{S_t}{t}-\theta\right)^2\Big|\theta\right]\right]
=\frac{1}{h^2}\Exp\left[\frac{\theta(1-\theta)}{t}\right]\leq\frac{1}{4 h^2t},
\label{eq:A15}
\end{multline}
where we used the fact that $\frac{S_t}{t}$ is not the optimum Bayes estimator of the mid-point, then we applied the Chebyshev's inequality, then the fact that $\frac{S_t}{t}$ is an estimator of $\theta$ with estimation error variance equal to $\frac{\theta(1-\theta)}{t}$ and finally that $\theta(1-\theta)\leq\frac{1}{4}$.
From \eqref{eq:A15} we conclude that $\Pro(|\hat{\vartheta}_t-\theta|> h)\to0$ as $t\to\infty$ therefore, there exists $N_\alpha$ such that $\Pro(|\hat{\vartheta}_{N_\alpha}-\theta|> h)<\alpha$.

Fix $N\geq N_{\alpha}$ and denote
$
\V_t(S_t,c)=\inf_{t\leq T\leq N}\Exp[c(T-t)+\C_T|\F_t],
$
where we underline the dependence of $\V_t$ on $c$ (in addition to $S_t$). For $0\leq c_1\leq c_2$ and $T\geq t$ we can write
$$
c_1(T-t)+\C_T\leq c_2(T-t)+\C_T,
$$
which, after taking expectation conditioned on $\F_t$ and then infimum over $t\leq T\leq N$, proves
that $\V_t(S_t,c)$ is increasing in $c$. The increase of $\V_t(S_t,c)$ with respect to $c$ also suggests that the optimum stopping time $\T(c)$, defined in \eqref{eq:opt-stop}, is a decreasing function of $c$. 

Consider now the sequence of optimum complementary coverage probabilities $\{\C_t\}$, we observe
\begin{multline}
\C_t=\inf_{\hat{\theta}}\Pro(|\hat{\theta}-\theta|>h|\F_t)=\inf_{\hat{\theta}}\Exp[\Pro(|\hat{\theta}-\theta|>h|\F_{t+1})|\F_t]\\
\geq \Exp\left[\inf_{\hat{\theta}}\Pro(|\hat{\theta}-\theta|>h|\F_{t+1})|\F_t\right]
=\Exp[\C_{t+1}|\F_t].
\label{eq:A10}
\end{multline}
We can then write
\begin{multline}
\Pro(|\hat{\vartheta}_{\T(c)}-\theta|>h)=\Exp[\C_{\T(c)}]=\C_0-\Exp\left[\sum_{t=0}^{\T(c)-1}\{\C_t-\C_{t+1}\}\right]\\
=\C_0-\Exp\left[\sum_{t=0}^{N}\{\C_{t}-\C_{t+1}\}\ind{\T(c)>t}\right]
=\C_0-\Exp\left[\sum_{t=0}^{N}\{\C_t-\Exp[\C_{t+1}|\F_t]\}\ind{\T(c)>t}\right],
\end{multline}
where for the last equality we used the fact that $\ind{\T(c)>t}$ is $\F_t$-measurable. This combined with \eqref{eq:A10} and the decrease of $\T(c)$ with respect to $c$, implies that $\Pro(|\hat{\vartheta}_{\T(c)}-\theta|>h)$ is increasing in $c$.

For $c=1$ we stop at 0 and, therefore, $\Pro(|\hat{\vartheta}_{\T(1)}-\theta|>h)=\Pro(|\hat{\vartheta}_{0}-\theta|>h)>\alpha$. Set now $c=0$ which suggests that the cost of sampling is zero and therefore the optimum is to stop at $N$ (we also deduce this by combining \eqref{eq:opt-stop} and \eqref{eq:A10}). This yields $\Pro(|\hat{\vartheta}_{\T(0)}-\theta|>h)=\Pro(|\hat{\vartheta}_{N}-\theta|>h)=\Exp[\C_N]$.
Now from \eqref{eq:A10} by averaging we conclude that $\Exp[\C_t]$ is decreasing in $t$ and for $N>N_\alpha$ we have $\Exp[\C_N]\leq\Exp[\C_{N_\alpha}]<\alpha$, implying $\Pro(|\hat{\vartheta}_{\T(0)}-\theta|>h)<\alpha$. As mentioned, $\Pro(|\hat{\vartheta}_{\T(c)}-\theta|>h)$ is increasing in $c$, if it is also continuous then there exists $0<c_*<1$ satisfying $\Pro(|\hat{\vartheta}_{\T(c_*)}-\theta|>h)=\alpha$ which means that $\T(c_*)$ solves the constrained problem. In case the function $\Pro(|\hat{\vartheta}_{\T(c)}-\theta|>h)$ exhibits a jump at $c_*$ such that for $c_*\text{-}$ the probability is strictly smaller than $\alpha$ while for $c_*\text{+}$ it is strictly larger, then before taking any samples we need to perform a randomization to decide which of the two stopping times $\T(c_*\text{-}),\T(c_*\text{+})$ to use. The randomization probability must be selected so that we satisfy the constraint with equality.\qed

\vskip0.2cm
\noindent\textit{Proof of Lemma\,\ref{lem:3}:}
We prove \eqref{eq:lem3-1} first. Set $\Omega_N=\varnothing$, i.e.~we stop necessarily at $N$. Then we note that
$$
T=\sum_{t=0}^{N-1}\ind{T>t}=(1+\cdots(1+\ind{S_{N-1}\in\Omega_{N-2}}(1+\ind{S_{N-1}\in\Omega_{N-1}}(1+\ind{S_{N}\in\Omega_{N}})))\cdots)
$$
suggesting that
$$
\Exp[T|\theta]=\Exp[(1+\cdots\Exp[(1+\ind{S_{N-1}\in\Omega_{N-1}}\Exp[(1+\ind{S_N\in\Omega_N})|\F_{N-1},\theta)|\F_{N-2},\theta]\cdots)|\theta].
$$
If we set $U_N(S_N)=0$ then we can define the backward recursion
\begin{multline*}
U_t(S_t)=\Exp[1+\ind{S_{t+1}\in\Omega_{t+1}}U_{t+1}(S_{t+1})|\F_t]=1+\Exp[\ind{S_{t+1}\in\Omega_{t+1}}U_{t+1}(S_{t+1})|\F_t]\\
=1+\Pro(X_{t+1}=1|S_t,\theta)\ind{S_{t}+1\in\Omega_{t+1}}U_{t+1}(S_t+1)+\Pro(X_{t+1}=0|S_t,\theta)\ind{S_{t}\in\Omega_{t+1}}U_{t+1}(S_t)\\
=1+\theta\ind{S_{t}+1\in\Omega_{t+1}}U_{t+1}(S_t+1)+(1-\theta)\ind{S_{t}\in\Omega_{t+1}}U_{t+1}(S_t),
\end{multline*}
which proves \eqref{eq:lem3-1} and, also, that $U_0(S_0)=\Exp[T|\theta]$. For \eqref{eq:lem3-2} we proceed similarly the only difference being that 
$\Pro(X_{t+1}=1|S_t)=g_{t+1}(S_t)$ with this probability being defined in \eqref{eq:gt}.

For \eqref{eq:lem3-3} and \eqref{eq:lem3-4} we follow similar steps. We have
\begin{multline*}
\ind{|\hat{\theta}_T-\theta|>h}=\sum_{t=0}^N\ind{|\hat{\theta}_t-\theta|> h}\ind{T=t}=\sum_{t=0}^N\ind{|\hat{\theta}_t-\theta|> h}\ind{S_t\not\in\Omega_t}\prod_{j=0}^{t-1}\ind{S_j\in\Omega_j}\\
=(\ind{|\hth_0-\theta|>h}\ind{S_0\not\in\Omega_0})+(\ind{|\hth_1-\theta|>h}\ind{S_1\not\in\Omega_1})\ind{S_0\in\Omega_0}+\cdots+(\ind{|\hth_N-\theta|>h}\ind{S_N\not\in\Omega_N})\prod_{j=1}^{N-1}\ind{S_j\in\Omega_j}.
\end{multline*}
Applying expectation given $\theta$ yields
\begin{multline*}
\Pro(|\hat{\theta}_T-\theta|>h|\theta)=\Exp[\ind{|\hth_0-\theta|}\ind{S_0\not\in\Omega_0}+\cdots\\
+\Exp[\ind{|\hth_{N-1}-\theta|>h}\ind{S_{N-1}\not\in\Omega_N}+(\Exp[\ind{|\hth_N-\theta|>h}\ind{S_N\not\in\Omega_N}|\F_{N-1},\theta])\ind{S_{N-1}\in\Omega_{N-1}}|\F_{N-2},\theta])\cdots|\theta].
\end{multline*}
Defining $W_N(S_N)=\ind{|\hth_N-\theta|>h}$ it is straightforward to see that the recursion in 
\eqref{eq:lem3-3} computes the desired complementary coverage probability. Similarly for \eqref{eq:lem3-4} only now instead of conditioning with respect to both $\F_t$ and $\theta$ we condition only with respect to $\F_t$. This concludes the proof.\qed

\vskip0.2cm
\noindent\textit{Proof of Theorem\,\ref{th:2}:} Let us first find upper and lower bounds of $\C_t(S_t)$ that are independent from $S_t$. From \cite[Theorem\,2.1]{Marchal} and for a random variable $X$ with density $\text{Beta}(x,p,q)$ we have that
\begin{equation}
\Exp[e^{\lambda(X-\mu)}]\leq e^{\frac{\lambda^2}{8(p+q+1)}},~\lambda>0,
\end{equation}
where $\mu=\frac{p}{p+q}$ is the average under the Beta density. Using the Markov inequality we can then write
\begin{multline}
\Pro(|X-\mu|>h)=\Pro(X-\mu>h)+\Pro(X-\mu<-h)=\Pro(X-\mu>h)+\Pro(1-X-(1-\mu)>h)\\
\leq \frac{\Exp[e^{\lambda(X-\mu)}]}{e^{\lambda h}}+\frac{\Exp[e^{\lambda(1-X-(1-\mu))}]}{e^{\lambda h}}\leq 2e^{\frac{\lambda^2}{8(p+q+1)}-\lambda h},
\label{eq:A3-1}
\end{multline}
where we used the fact that if $X$ is Beta distributed with parameters $p,q$ then $1-X$ is also Beta with parameters $q,p$. Selecting in \eqref{eq:A3-1} $\lambda=4(p+q+1)h$ yields the tightest upper bound, namely
\begin{equation}
\Pro(|X-\mu|>h)\leq 2e^{-2h^2(p+q+1)}.
\label{eq:A4}
\end{equation}
We can now use this result to upper bound $\C_t(S_t)$. We observe that
\begin{equation}
\C_t(S_t)=\inf_{\hat{\theta}_t}\Pro(|\theta-\hat{\theta}_t|>h|\F_t)\leq\Pro\big(|\theta-\Exp[\theta|\F_t]|>h|\F_t\big)\leq 2e^{-2h^2(t+2a+1)}.
\label{eq:A5}
\end{equation}
For for the last inequality we used \eqref{eq:A4} and the fact that $\theta$ given $\F_t$ is Beta distributed with parameters $p=S_t+a$ and $q=t-S_t+a$.

Let us now find a lower bound for $\C_t(S_t)$. From \cite[Page 944, Formula 26.5.15]{Abramowitz} we conclude that $I_x(p,q)>I_x(p+1,q-1)$ for $q>1$. Using this inequality repeatedly in \eqref{eq:25} we conclude
\begin{multline}
\Pro(|\hat{\theta}_t-\theta|>h|\F_t)=1-I_{\min\{1,\hth_t+h\}}(S_t+a,t-S_t+a)+I_{\max(0,\hth_t-h)}(S_t+a,t-S_t+a)\\
=I_{\max\{0,1-h-\hth_t\}}(t-S_t+a,S_t+a)+I_{\max(0,\hth_t-h)}(S_t+a,t-S_t+a)\\
\geq I_{\max\{0,1-h-\hth_t\}}(t+2n_a+\delta_a,\delta_a)+I_{\max\{0,\hth_t-h\}}(t+2n_a+\delta_a,\delta_a),
\label{eq:A40}
\end{multline}
where for the second equality we used the property $1-I_x(p,q)=I_{1-x}(q,p)$ and where $n_a,\delta_a$ are defined as
$$
n_a=\left\{\begin{array}{cl}[a]&\text{if}~a~\text{not an integer}\\a-1&\text{if}~a~\text{an integer},\end{array}\right.~~
\delta_a=\left\{\begin{array}{cl}a-[a]&\text{if}~a~\text{not an integer}\\1&\text{if}~a~\text{an integer},\end{array}\right.
$$
where $[a]$ denotes integer part of $a$.
Since $a>0$ we have $n_a\geq0$, $1\geq\delta_a>0$ and $a=n_a+\delta_a$.
By taking the derivative of the last sum in \eqref{eq:A40} with respect to $\hth_t$ we can show that it has the same sign as the following expression
$$
\phi(\hth_t)=\frac{(\hth_t-h)^{t+2n_a+\delta_a-1}}{(1+h-\hth_t)^{1-\delta_a}}-\frac{(1-h-\hth_t)^{t+2n_a+\delta_a-1}}{(\hth_t+h)^{1-\delta_a}}.
$$
Now it is easy to verify that $\phi(1-\hth_t)=-\phi(\hth_t)$ therefore it is sufficient to analyze the sign of $\phi(\hth_t)$ for $h\leq\hth_t\leq0.5$. When $t\geq1$ and because $1\geq\delta_a$ we can see that the sign is negative for any value of $a$, suggesting that we have a minimum for $\hth_t=0.5$. Therefore, if $\Gamma(x)$ denotes the Gamma function, then for $t\geq1$ we can write
\begin{multline}
\C_t\geq 2I_{0.5-h}(t+2n_a+\delta_a,\delta_a)\geq 2\frac{\Gamma(t+2n_a+2\delta_a)(0.25-h^2)^{\delta_a}}{\Gamma(t+2n_a+\delta_a+1)\Gamma(\delta_a)}(0.5-h)^{t+2n_a}\\
= 2\frac{\Gamma(t+2n_a+2\delta_a+1)(0.25-h^2)^{\delta_a}}{(t+2n_a+2\delta_a)\Gamma(t+2n_a+\delta_a+1)\Gamma(\delta_a)}(0.5-h)^{t+2n_a}\\
\geq2\frac{(0.25-h^2)^{\delta_a}}{(t+2n_a+2\delta_a)\Gamma(\delta_a)}(0.5-h)^{t+2n_a}.
\label{eq:A30}
\end{multline}
In the previous expression the second inequality comes from \cite[Page 944, Formula 26.5.16]{Abramowitz}; for the next equality we used the property $\Gamma(x+1)=x\Gamma(x)$; while for the last inequality we used the  increase of $\Gamma(x)$ for $x\geq1.5$, which is true in our case for $t\geq1$ and any $a>0$.

Having established bounds for $\C_t$ we can now compute an upper bound $N$ for $t_{\rm up}$ and a lower bound $\nu$ for $t_{\rm lo}$ therefore proving their existence and demonstrating properties i) and ii). We first note that if $\C_N\leq c$ in \eqref{eq:backwards2} we will have $\C_N\leq c+\tilde{\V}_N$ meaning that $\V_N=\C_N$ and consequently $N$ is a stopping instant for \textit{all} values of $S_t$. This implies that $\T\leq N$. Quantity $t_{\rm up}$ is the smallest $N$ for which this inequality is true for all $S_t$. Requiring $2e^{-2h^2(N+2a+1)}\leq c$ we obtain
\begin{equation}
N=\left\lceil\max\left\{0,\frac{|\log(c)|+\log(2)}{2h^2}-2a-1\right\}\right\rceil.
\label{eq:N}
\end{equation}

To find a lower bound $\nu$ for $t_{\rm lo}$ we combine the lower bound of $\C_t$ with an upper bound for $\V_t$. Finding the latter is straightforward. Indeed if we start from time instant $N$ which, as we argued, is selected so that $\C_N\leq c$, then using induction and the fact that
$$
\V_t=\min\{\C_t,c+\Exp[V_{t+1}|\F_t]\}\leq c+\Exp[\V_{t+1}|\F_t]
$$
we can show that $\V_t\leq c+c(N-t)=c(N+1-t)$. It is then clear that, as long as $c(N+1)\leq\C_0$, for any $t\geq 1$ for which we have
\begin{equation}
c(N+1-t)(t+2n_a+2\delta_a)\leq 2\frac{(0.25-h^2)^{\delta_a}}{\Gamma(\delta_a)}(0.5-h)^{t+2n_a}
\label{eq:A31}
\end{equation}
we do not stop at this time instant. In fact we can see that we have an interval of the form $t\in[0,\ldots,\nu]$ during which no stopping can occur. A rough estimate of $\nu$ can be obtained by solving instead of \eqref{eq:A31} the simpler alternative $\max_tc(N+1)(t+2n_a+2\delta_a)=\frac{c}{4}(N+2n_a+2\delta_a)^2\leq\frac{2(0.25-h^2)^{\delta_a}}{\Gamma(\delta_a)}(0.5-h)^{\nu+2n_a}$ which yields
\begin{equation}
%
\nu=\left\lfloor\max\left\{0,\frac{|\log(c)|-\log\left((N+2n_a+\delta_a)^2\Gamma(\delta_a)\right)+ \log\left(8(0.25-h^2)^{\delta_a}\right)}{|\log(0.5-h)|}-2n_a\right\}\right\rfloor,
\label{eq:nu}
\end{equation}
provided $c$ satisfies $c\leq\frac{\C_0}{N+1}$. Regarding the latter, if we are in the non-trivial case where we do not stop at time 0 then $\alpha<\C_0$, consequently it is sufficient to have $c\leq\frac{\alpha}{N+1}$. We thus conclude that for small enough $c$ there is a lower limit $t_{\rm lo}\geq \nu$ which is nontrivial. This concludes the proof.\qed

\ifCLASSOPTIONcaptionsoff
  \newpage
\fi

\bibliographystyle{IEEEtran}
\bibliography{ref}

\begin{thebibliography}{10}
\providecommand{\url}[1]{#1}
\csname url@samestyle\endcsname
\providecommand{\newblock}{\relax}
\providecommand{\bibinfo}[2]{#2}
\providecommand{\BIBentrySTDinterwordspacing}{\spaceskip=0pt\relax}
\providecommand{\BIBentryALTinterwordstretchfactor}{4}
\providecommand{\BIBentryALTinterwordspacing}{\spaceskip=\fontdimen2\font plus
\BIBentryALTinterwordstretchfactor\fontdimen3\font minus
  \fontdimen4\font\relax}
\providecommand{\BIBforeignlanguage}[2]{{%
\expandafter\ifx\csname l@#1\endcsname\relax
\typeout{** WARNING: IEEEtran.bst: No hyphenation pattern has been}%
\typeout{** loaded for the language `#1'. Using the pattern for}%
\typeout{** the default language instead.}%
\else
\language=\csname l@#1\endcsname
\fi
#2}}
\providecommand{\BIBdecl}{\relax}
\BIBdecl

\bibitem{sullivan13}
A.~Sullivan, D.~Raben, J.~Reekie, M.~Rayment, A.~Mocroft \emph{et~al.},
  ``Feasibility and effectiveness of indicator condition-guided testing for
  \uppercase{HIV}: results from \uppercase{HIDES I} (\uppercase{HIV} indicator
  diseases across \uppercase{E}urope study),'' \emph{PLoS ONE}, vol.~8, no.~1,
  p. e52845, 2013.

\bibitem{abramson13}
J.~Abramson, R.~Takvorian, D.~Fisher, Y.~Feng, E.~Jacobsen \emph{et~al.},
  ``Oral clofarabine for relapsed/refractory non-\uppercase{H}odgkin lymphomas:
  results of a phase 1 study.'' \emph{Leukemia \& Lymphoma}, vol.~54, pp.
  1915--1920, 2013.

\bibitem{mk98}
J.~T. Morisette and S.~Khorram, ``Exact binomial confidence interval for
  proportions,'' \emph{Photogrammetric Engineering \& Remote Sensing}, vol.~64,
  no.~4, pp. 281--283, 1998.

\bibitem{ac98}
A.~Agresti and B.~A. Coull, ``Approximate is better than ``exact'' for interval
  estimation of binomial proportions,'' \emph{The American Statistician},
  vol.~52, pp. 119--126, 1998.

\bibitem{bcd01}
L.~D. Brown, T.~Cai, and A.~Dasgupta, ``Interval estimation for a binomial
  proportion (with discussion),'' \emph{Statistical Science}, vol.~16, pp.
  101--133, 2001.

\bibitem{n98}
R.~G. Newcombe, ``Two-sided confidence intervals for the single proportion:
  comparison of seven methods,'' \emph{Statistics in Medicine}, vol.~17, pp.
  857--872, 1998.

\bibitem{v93}
S.~E. Vollset, ``Confidence intervals for a binomial proportion,''
  \emph{Statistics in Medicine}, vol.~12, no.~9, pp. 809--824, 1993.

\bibitem{w27}
E.~B. Wilson, ``Probable inference, the law of succession, and statistical
  inference,'' \emph{Journal of the American Statistical Association}, vol.~22,
  pp. 209--212, 1927.

\bibitem{cp34}
C.~J. Clopper and E.~S. Pearson, ``The use of confidence or fiducial limits
  illustrated in the case of the binomial,'' \emph{Biometrika}, vol.~26, pp.
  404--413, 1934.

\bibitem{s54}
T.~E. Sterne, ``Some remarks on confidence of fudicial limits,''
  \emph{Biometrika}, vol.~41, no.~1, pp. 275--278, 1954.

\bibitem{c56}
E.~L. Crow, ``Confidence intervals for a proportion,'' \emph{Biometrika},
  vol.~43, pp. 423--435, 1956.

\bibitem{bs83}
C.~R. Blyth and H.~A. Still, ``Binomial confidence intervals,'' \emph{Journal
  of the American Statistical Association}, vol.~78, pp. 108--116, 1983.

\bibitem{bcd02}
L.~D. Brown, T.~Cai, and A.~DasGupta, ``Interval estimation for a binomial
  proportion and asymptotic expansions,'' \emph{The Annals of Statistics},
  vol.~30, pp. 160--201, 2002.

\bibitem{r03}
J.~Reiczigel, ``Confidence intervals for the binomial parameter: some new
  considerations,'' \emph{Statistics in Medicine}, vol.~22, no.~4, pp.
  611--621, 2003.

\bibitem{pa08}
A.~M. Pires and C.~Amado, ``Interval estimators for a binomial proportion:
  comparison of twenty methods,'' \emph{REVSTAT - Statistical Journal}, vol.~6,
  no.~2, pp. 165--197, 2008.

\bibitem{t14}
M.~Thulin, ``The cost of using exact confidence intervals for a binomial
  proportion,'' \emph{Electron. J. Statist.}, vol.~8, no.~1, pp. 817--840,
  2014.

\bibitem{Jegourel17}
C.~Jegourel, J.~Sun, and J.~S. Dong, \emph{Sequential Schemes for Frequentist
  Estimation of Properties in Statistical Model Checking}.\hskip 1em plus 0.5em
  minus 0.4em\relax Springer International Publishing, 2017, pp. 333--350.

\bibitem{cr65}
Y.~S. Chow and H.~Robbins, ``On the asymptotic theory of fixed-width sequential
  confidence intervals for the mean,'' \emph{The Annals of Mathematical
  Statistics}, vol.~36, no.~2, pp. 457--462, 1965.

\bibitem{t61}
M.~Tanaka, ``On a confidence interval of given length for the parameter of the
  binomial and the \uppercase{P}oisson distributions,'' \emph{Annals of the
  Institute of Statistical Mathematics}, vol.~13, pp. 201--215, 1961.

\bibitem{f10}
J.~Frey, ``Fixed-width sequential confidence intervals for a proportion,''
  \emph{The American Statistician}, vol.~64, pp. 242--249, 2010.

\bibitem{c77}
P.~Cabilio, ``Sequential estimation in \uppercase{B}ernoulli trials,''
  \emph{The Annals of Statistics}, vol.~5, no.~2, pp. 342--356, 1977.

\bibitem{p88}
H.~V. Poor, \emph{An Introduction to Signal Detection and Estimation},
  2nd~ed.\hskip 1em plus 0.5em minus 0.4em\relax New York: Springer, 1988.

\bibitem{Shiryaev}
A.~N. Shiryaev, \emph{Optimal Stopping Rules}.\hskip 1em plus 0.5em minus
  0.4em\relax Springer, 1978.

\bibitem{choi}
T.~Choi and R.~V. Ramamoorthi, \emph{Remarks on consistency of posterior
  distributions}, ser. Collections.\hskip 1em plus 0.5em minus 0.4em\relax
  Beachwood, Ohio, USA: Institute of Mathematical Statistics, 2008, vol. Volume
  3, pp. 170--186.

\bibitem{Abramowitz}
M.~Abramowitz and I.~A. Stegun, \emph{Handbook of Mathematical Functions},
  9th~ed.\hskip 1em plus 0.5em minus 0.4em\relax Dover Publication, 1972.

\bibitem{gym17}
G.~V. Moustakides, T.~Yaacoub, and Y.~Mei, ``Sequential estimation based on
  conditional cost,'' \emph{Proceedings of IEEE International Symposium on
  Information Theory}, pp. 436--440, June 2017.

\bibitem{Marchal}
O.~Marchal and J.~Arbel, ``On the sub-\uppercase{G}aussianity of the
  \uppercase{B}eta and \uppercase{D}irichlet distributions,'' \emph{Electron.
  Commun. Probab.}, vol.~22, no.~54, pp. 1--14, 2017.

\end{thebibliography}

\end{document}